\documentclass[twocolumn]{autart}    % 

\usepackage{times}
\usepackage{amsmath, amssymb,amsfonts}
\usepackage{algorithmic}
\usepackage{graphicx}
\usepackage{algorithm,algorithmic}
\usepackage{textcomp}
\usepackage{booktabs}
\usepackage{siunitx}
\usepackage[dvipsnames]{xcolor}
\usepackage{nicefrac}
\usepackage{comment}
\usepackage{cite}
\usepackage{siunitx}
\usepackage{mathtools}

\newtheorem{theorem}{Theorem}  % thm

\newtheorem{lemma}{Lemma}  % lem
\newtheorem{proposition}{Proposition}  % prop

\newtheorem{assumption}{Assumption}  % assum

\allowdisplaybreaks
\newcommand{\xstar}[1]{x^*_{#1}(a)}

\newcommand{\xstartilde}[1]{\tilde{x}^*_{#1}}

\newlength{\mlLegendThickness}
\newlength{\mlLegendHeight}
\newcommand{\mlLineLegend}[1]{\mbox{\color{#1}
		\protect\rule[\mlLegendHeight]{3mm}{\mlLegendThickness}\hspace*{-1mm}
}}
\newcommand{\mlLineLegendDashed}[1]{\mbox{\color{#1}
\protect\rule[\mlLegendHeight]{1.5mm}{\mlLegendThickness}\hspace*{0mm}
\protect\rule[\mlLegendHeight]{1.5mm}{\mlLegendThickness}\hspace*{-1mm}
}}

\definecolor{dunkelblau}{rgb}{0.0, 0.2314, 0.6196}%
\definecolor{hellblau}{rgb}{0.000, 0.7451, 1.0000}%
\definecolor{rot}{rgb}{0.6980,0.1333,0.1333}%

\newcommand{\mycomment}[1]{}

\newcommand{\old}[1]{}

\newcommand{\reviewed}[1]{\textcolor{black}{#1}}

\newcommand{\new}[1]{\textcolor{black}{#1}}

\newcommand{\rev}[1]{\textcolor{black}{#1}}

\begin{document}

\begin{frontmatter}

\title{Stabilization of Age-Structured Competing Populations\thanksref{footnoteinfo}} % 

\thanks[footnoteinfo]{
% TBD This paper was not presented at any IFAC meeting. 
Corresponding author: Carina Veil.}

\author[stanford,isys]{Carina Veil}\ead{cveil@stanford.edu},    % Add the 
\author[ucsd]{Miroslav Krstic}\ead{krstic@ucsd.edu}, % e-mail address 
\author[sdsu]{Patrick McNamee}\ead{pmcnamee5123@sdsu.edu},
\author[isys]{Oliver Sawodny}\ead{sawodny@isys.uni-stuttgart.de}  % (ead) as shown

\address[stanford]{Department of Mechanical Engineering, Stanford University, Stanford, CA 94305, USA}                                        
\address[ucsd]{Department of Mechanical and Aerospace Engineering, University of California, San Diego, La Jolla, CA 92093-0411, USA}             
\address[sdsu]{Department of Mechanical Engineering, San Diego State University, San Diego, CA 92182, USA}          
\address[isys]{Institute for System Dynamics, University of Stuttgart, 70563 Stuttgart, Germany} 
          
% \begin{keyword}      
% \end{keyword}                             
\begin{abstract}    % Abstract of not more than 200 words.
Age-structured models \rev{capture} the dynamic behavior of populations over time and result in nonlinear integro-partial differential equations (IPDEs). \rev{These processes arise in various fields such as biotechnology, economics, or demography. While coupled age-structured IPDEs modeling two or more interacting species occur naturally in epidemiology and ecology, they remain relatively underexplored.}
\rev{Prior work has primarily addressed stable and marginally stable dynamics. In constrast,} this work considers an exponentially unstable model of two competing predator populations, formally referred to in the literature as ``competition'' dynamics. 
If one were to apply an input that simultaneously harvests both predator species, one would have control over only the product of the densities of the species, not over their ratio. 
Therefore, it is necessary to design a control input that directly harvests only one of the two predator species, while indirectly influencing the other via a backstepping approach. 
The model is transformed into a system of two coupled ordinary differential equations (ODEs), of which only one is actuated, and two autonomous, exponentially stable integral delay equations (IDEs) which enter the ODEs as nonlinear disturbances. The ODEs are globally stabilized with backstepping and an estimate of the region of attraction of the asymptotically stabilized equilibrium of the full IPDE system is provided, under a positivity restriction on control. \rev{Additionally, the full IPDE system is also shown to be local exponential stable.}
Such generalizations \rev{of competition dynamics} open exciting possibilities for future research directions \rev{for systems with more than two species}.
\end{abstract}

\end{frontmatter}

\allowdisplaybreaks

\section{Introduction}

Age-structured population models capture the dynamics of renewable populations over time and serve to explain past population fluctuations or predict future growth. Being structured around age cohorts, these models lead to nonlinear integro-partial differential equations (IPDEs) \new{with positive states}, \new{capturing sub-population dynamics governed by age-dependent contact rates \cite{heesterbeek2005law} and  and shaped by birth and death rates varying with age and, potentially, time.}

\new{Such models have been applied across biology, biotechnology, and demography \cite{inaba2017age, martcheva2015introduction, brauer2012mathematical, rong2007mathematical, gyllenberg1983stability} to link population dynamics with socio-economic outcomes; for instance, forecasting workforce and healthcare demand \cite{freiberger2024optimization, tahvonen2009economics}, understanding the spread of behaviors in carceral systems \cite{ibrahim2022mathematical, sooknanan2023criminals}, or guiding equitable educational strategies under demographic shifts \cite{lutz2005toward}.}
\new{On a smaller scale, the dynamics of interacting populations are studied through chemostat models, where fresh nutrient solution is fed to a bioreactor at the same rate as it is extracted. To achieve the desired amount of biomass within the reactor, this rate acts as control input.}
% , making harvesting and dilution synonymous.}
From a control perspective, most models involving interacting populations, such as in protein synthesis \cite{li2008competition} or wastewater treatment \cite{andrews1974dynamic}, are formulated with ordinary differential equations (ODEs), but structured models are increasingly relevant in ecology and epidemiology \cite{holmes1994partial, nguyen2020analysis, smith2013chemostats}.
In these systems, actuation is inherently subtractive: species can only be removed, not added, through mechanisms such as dilution in chemostats, harvesting in ecosystems, or vaccination in epidemics.

\textbf{Feedback control related work.}
Research on infinite-dimensional population models has mainly focused on open-loop stability under parameter restrictions \cite{inaba1990threshold, liu2015global} or on optimal control formulations \cite{albi2021control, feichtinger2003optimality}. 
\rev{For the stabilization of age-structured systems, the work of Karafyllis and Krsti\'c
\cite{karafyllis2017stability} laid the foundation. Their study considered a \textbf{single-population} model without competition terms, which is inherently stable, and introduced fundamental Lyapunov-based methods that have since become standard tools for feedback design. Such single population models are, for example, used to control the dilution rate of the chemostat such that the biomass follows certain trajectories  \cite{schmidt2018yield, kurth2021tracking, kurth2023control}.
% \cite{haacker2024stabilization, schmidt2018yield, kurth2021tracking, kurth2021optimal, kurth2023control}.
In contrast, \textbf{multi-population} IPDE models remain underexplored. Feedforward control of a stable system with two competing species was designed in \cite{kurth2022model}. Veil et al. \cite{veil2024stabilization} addressed a \textbf{two-population predator–prey} model with interaction terms that render the equilibrium marginally stable. There, both populations share a control input, corresponding to simultaneous harvesting.}

\begin{figure}
    \centering
   \includegraphics[width=0.9\columnwidth]{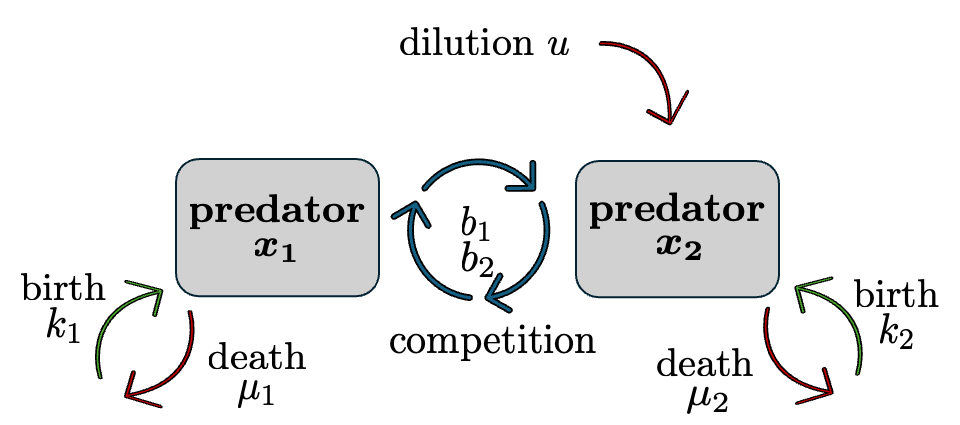}
    \caption{The two predators $x_1$ and $x_2$ compete via the terms $b_i$. Each species is affected by mortality $\mu_i$ and new population can only enter the system through birth $k_i$. The dilution input $u$ is selective and has a repressive effect on only one species.
    % : By directly actuating one species and employing a backstepping strategy to influence the other indirectly, both populations can be stabilized.
    }
    \label{fig:ppmodel}
\end{figure}

\textbf{Contribution.}
\rev{The present work addresses a more challenging and qualitatively different scenario: a \textbf{two-population predator–predator} model with exponentially unstable open-loop dynamics \rev{(cf. Section~\ref{sec:control-design})} and only one non-extinction equilibrium. In this setting, shared actuation cannot stabilize the system  because it can only influence the product (but not the ratio) of the species densities. Hence, stabilization requires a discriminative actuation, where the control is applied selectively to one of the two predator populations. By directly actuating one species and employing a backstepping strategy to influence the other indirectly, both populations can be globally stabilized.}
% This paper addresses the feedback stabilization of such an IPDE system with positive states and input, modeling two competing predator populations. In contrast to previous models, the system is exponentially unstable with only one non-extinction equilibrium, which cannot be stabilized by simultaneous harvesting alone;
% %In the presented predator-predator case, simultaneous harvesting
% the latter only enables control over the product of the species densities but not their ratio. Hence, it is necessary to discriminate via actuation -- that is, to apply input selectively to one of the two predator populations. By directly actuating one species and employing a backstepping strategy to influence the other indirectly, both populations can be stabilized.
\rev{Building on the transformation framework of \cite{karafyllis2017stability}, the proposed approach reformulates the IPDE system into two coupled ODEs and two autonomous but exponentially stable integral delay equations.
A control Lyapunov functional (CLF) and feedback law is constructed that need to differ fundamentally from those in \cite{karafyllis2017stability, veil2024stabilization}, both in structure and in the resulting region of attraction, due to the system’s exponential instability and asymmetric actuation. The resulting control design employs an unconventional backstepping technique.
% , since the ODE describing the first species is not affine in the virtual control variable representing the second species. 
The main similarity among all three works lies in the construction of the Lyapunov functionals \cite{karafyllis2017stability}. As such, this work introduces a benchmark-level problem for stabilization of multi-species age-structured systems under exponential instability and actuation asymmetry.}

% With a backstepping approach, a feedback is designed which employs possibly negative harvesting for global stabilization of the ODE model, while guaranteeing regional regulation with positive harvesting. 
% It is established that, though backstepping control is derived solely from the ODE dynamics, it achieves regional stabilization in the presence of the internal dynamics of the transformed ODE-IDE system and, hence, for the two coupled IPDEs.  
% The resulting control Lyapunov function (CLF), control law, and even the approach to estimating the region of attraction differ fundamentally from \cite{karafyllis2017stability, veil2024stabilization}. The significantly more elaborate control design employs an unconventional form of backstepping. As such, this work introduces a benchmark-level problem for stabilization of multi-species age-structured systems under exponential instability and actuation asymmetry.

\begin{table}[t]
\centering
\caption{\rev{Comparison of related previous work.}}
\label{tab:related-work}
\resizebox{\columnwidth}{!}{%
\begin{tabular}{lcccc}
\toprule
Work & Species & Interaction & Stability & Actuation \\
\midrule
\cite{karafyllis2017stability} & 1 & -  & Stable & -  \\
% \cite{kurth2022model} & 2 & -  & Stable & Shared  \\
\cite{veil2024stabilization} & 2 & Predator--Prey & Marginally Stable & Shared \\ 
\midrule
Here & 2 & Predator--Predator & Exponentially Unstable & Selective \\
\bottomrule
\end{tabular}}
\end{table}

\new{\emph{Organization.} Section~\ref{sec:system-model} introduces the model, its equilibria and the system transformation. In Section~\ref{sec:control-design}, a simplified ODE system is considered to derive a backstepping control law which globally stabilizes the ODE system. In Section~\ref{sec:stability-results}, it is established that this control law also locally stabilizes the infinite dimensional system, along with an explicit estimate of the region of attraction.
Finally, simulations are presented in Section~\ref{sec:simulations}.
% , and Section~\ref{sec:conclusion} concludes the article.
% \rev{\emph{This manuscript is an extended version of a shorter manuscript from the IEEE Conference on Decision and Control (2025) \cite{veil2025stabilization}, which contains elements of the analysis omitted in the conference version, and extensions regarding open-loop stability, local exponential stability, observability, and the implementation of the control.}}
}

\rev{{This article is an extended journal version of a significantly shorter conference paper  \cite{veil2025stabilization}, with the following added value: the results for the open-loop instability in Section~\ref{sec:control-design}, an additional exponential stability result given as Part 4 of Theorem~\ref{thm:local-stability}, and an exponential stability result of the linearization for a controller using lumped measurements instead of exact population densities in Proposition~\ref{prop:local-exponential-stability}.}}

%%%%%%%%%%%%%%%%%%%%%%%%%%%%%%%%%%%%%%%%%%%%%%%%%%%%%%%%
\section{Population Model}\label{sec:system-model}
The age-structured predator-predator model with initial conditions (IC) and boundary conditions (BC) is given by
\begin{subequations} \label{eq:sys}
    \begin{align}
	 x_1'(a,t) + \dot x_1(a,t)
	=& - x_1(a,t) \Big[ \mu_1(a) \notag \\
    & + \int_0^A b_1(\alpha) x_2(\alpha,t) d\alpha \Big] \label{eq:ipde1} \\
		 x_2'(a,t) + \dot x_2(a,t)
	=& - x_2(a,t)\Big[ \mu_2(a)\notag \\
    &  + \int_0^A b_2(\alpha) x_1(\alpha,t) d \alpha + u(t)\Big] \label{eq:ipde2} \\
    \text{IC}: \qquad\quad \ \ x_i(a, 0) =& \ x_{i,0}(a), \\
	\text{BC}: \qquad\quad \, \ x_i(0,t) =& \ \int_0^A k_i(a) x_i(a,t) d a  \label{eq:bc}
     \end{align} 
\end{subequations}
where, for $i,j \in \{1,2\}$, $i\neq j$, $x_i(a,t)>0$ is the population density, i.~e. the amount of organisms of a certain age $a \in [0,A]$ of the two interacting populations $x_1(a,t)$ and $x_2(a,t)$ with $\reviewed{(a,t) \in [0,A] \times \mathbb{R}_+}$, their derivatives $\dot{x}_i$ with respect to time and $x'_i$ with respect to age, and the constant maximum age $A>0$.
The interaction kernels $b_i(a):[0,A]\rightarrow\mathbb{R}_0^+$, the mortality rates $\mu_i(a):[0,A]\rightarrow\mathbb{R}_0^+$, the birth rates $k_i(a):[0,A]\rightarrow\mathbb{R}_0^+$ are functions with $\int_0^A \mu_i(a)d a > 0$, $\int_0^A b_i(a)d a > 0$, $\int_0^A k_i(a)d a >0$. The dilution rate $u(t):\mathbb{R}^+\rightarrow \mathbb{R}_0^+$, is an input affecting only one species. %(species-specific/species-discriminating harvesting). 
%\reviewed{Note that the input is positive at all times -- negative harvesting/dilution corresponds to adding population to the system.}

System~\eqref{eq:sys} describes predator-predator population dynamics where only one predator is harvested.
Importantly, \eqref{eq:sys} is unstable and its states are functions of $a\in[0,A]$ with values $x_i$ at time $t$, that belong to the function spaces $\mathcal{F}_i$, $i=1,2$, 
\begin{align}
    \mathcal{F}_i = \Big\{ \xi \in PC^1 ([0,A];(0,\infty)):% \nonumber \\ 
    \xi(0) =\int_0^A k_i(a) \xi(a) d a 
    \Big\}.
\end{align}
% ChatGPT Def shortened
% \rev{A function \(f : [0,A] \to \mathbb{R}\) is called piecewise continuous if there exist points
% \(0 = t_0 < t_1 < \cdots < t_N = A\) such that  
% (i) \(f\) is continuous on each open interval \((t_{i-1}, t_i)\), and  
% (ii) the one-sided limits \(\lim_{t\to t_{i-1}^+} f(t)\) and \(\lim_{t\to t_i^-} f(t)\) exist and are finite for all \(i\).
% The set of all piecewise continuous functions on \([0,A]\) is  
% $PC[0,A] := \{\, f : [0,A] \to \mathbb{R} \mid f \text{ satisfies (i) and (ii) above} \}$.}
% For any subset $\mathcal{R}\subseteq \mathbb{R}$ and for any $A>0$, $PC^1([0,A];\mathcal{R})$ denotes the class of all functions $f\in C^0([0,A];\mathcal{R})$ for which there exists a finite (or empty) set $B\subset(0,A)$ such that: (i) the derivative $x'(a)$ exists at every $a \in (0,A)\backslash B$ and is a continuous function on $(0,A) \backslash B$, (ii) all meaningful right and left limits of $x'(a)$ when $a$ tends to a point in $B \cup \{0,A\}$ exist and are finite. 
For any subset $\mathcal{R}\subseteq \mathbb{R}$ and for any $A>0$, $PC^1([0,A];\mathcal{R})$ denotes the class of all functions $f\in C^0([0,A];\mathcal{R})$ for which there exists a finite (or empty) set $B\subset(0,A)$ such that: \rev{(i) the derivative $x'(a)$ exists at every $a \in [0,A]\backslash B$ and is a continuous function on $[0,A] \backslash B$ and (ii) the limits $\lim_{a\to b^+} x'(a)$ and $\lim_{a\to b^-} x'(a)$ exists, are finite, and equal for all $b\in B$.}
\rev{For ease of notation, we let $x = [x_1,x_2]^\top$ and $\mathcal{F} = \mathcal{F}_1 \times \mathcal{F}_2$.}

\subsection{Equilibrium}
\new{To determine the equilibria of the population dynamics~\eqref{eq:sys}, the well-known Lotka-Sharpe condition is needed, which provides a mathematical criterion for the existence of population equilibria. The proof can be found in \cite{sharpe1911problem}.}

\begin{lemma}[Lotka-Sharpe condition \cite{sharpe1911problem}]
\label{lem:ls}
The equations
\begin{align}\label{eq:lotkasharpe}
    %    1&=\langle k_i,\tilde{x}_i^*\rangle=:\langle \tilde{k}_i,1\rangle
    \int_0^A \tilde{k}_i(a) d a = 1, \ i=1,2
    %\langle \tilde{k}_i, 1 \rangle = 1
\end{align}
with
\begin{align}
    \tilde{k}_i(a) =  k_i(a) e^{-\int_0^a \left(\mu_i(s) + \zeta_i \right)d s}
\end{align}
have unique real-valued solutions 
$\zeta_1(k_1,\mu_1)$ and $\zeta_2(k_2,\mu_2)$ \rev{that are constant and} 
%$\reviewed{\zeta_1}$ and $\reviewed{\zeta_2}$,  
depend on the birth rates $k_i$ and mortality rates $\mu_i$.
\end{lemma}

\new{Using the Lotka-Sharpe condition in Lemma~\ref{lem:ls}, the equilibrium of~\eqref{eq:sys} is stated.}
\begin{proposition}[Equilibrium]
\label{prop:steady}
The equilibrium  state \\ $(\xstar{1},\xstar{2})$ of the population system (\ref{eq:sys}) is given by
\begin{align}
    % x_i^* (a) &= x_i^*(0)\underbrace{e^{-{\int_0^a (\zeta_i + \mu_i(s)) d s}}}_{\tilde x_i^*(a)},\label{eq:ss_profiles}\\
    x_i^* (a) &= x_i^*(0) e^{-{\int_0^a (\zeta_i + \mu_i(s)) d s}} =: x_i^*(0) \tilde x_i^*(a),\label{eq:ss_profiles}\\
    u^* &=  \zeta_2 - \lambda_1 % \in \left(0, \zeta_2\right)\,,
    \label{eq:u_star_constraint}
\end{align}
with unique parameters 
%$\zeta_i(k_i,\mu_i)$ 
$\reviewed{\zeta_i}$
resulting from the Lotka-Sharpe condition of Lemma~\ref{lem:ls}, \reviewed{constant interaction terms}
\begin{align}
    \lambda_i(u^*) = \int_0^A b_j(a)x_i^*(a) \reviewed{da},\label{eq:def_lambda}
\end{align}
and the positive  concentrations of the newborns 
\begin{subequations}
    \begin{align}
    	x_1^*(0) &=  \frac{\zeta_2-u^*}{ \int_0^A b_2(a) \xstartilde{1} (a)d a} >0, \\
    	x_2^*(0) &= \frac{\zeta_1}{\int_0^A b_1(a)\xstartilde{2}(a)d a }  > 0\,.
    \end{align}
    \label{eq:ic_constraint}
\end{subequations}
\end{proposition}

\begin{pf}%[Proof of Proposition \ref{prop:steady}.]
\rev{Stationarity of (\ref{eq:sys})}
%Neglecting the time dependence in (\ref{eq:sys}) 
results in the constant dilution rate $u^*$ and interaction terms \reviewed{$\lambda_i$}, \reviewed{which are}
merged in the parameters \rev{$\zeta_1=\lambda_2$ and $\zeta_2=u^* +\lambda_1$},
\begin{subequations}\label{eq:ss}
\begin{align}
 0 &= -x_1^{*'} (a) - \xstar{1} \Big(\mu_1(a) + \rev{\zeta_1 }\Big),\label{eq:ss_ode1} \\
  0 &= -x_2^{*'} (a) - \xstar{2} \Big(\mu_2(a) + \rev{\zeta_2} \Big).\label{eq:ss_ode2}
\end{align}%
\end{subequations} %
The equilibrium profiles (\ref{eq:ss_profiles}) result from solving ODEs (\ref{eq:ss_ode1}), \eqref{eq:ss_ode2} for arbitrary initial conditions $ x_i^*(0)$. Inserting the solutions into the boundary conditions (\ref{eq:bc}) results in the Lotka-Sharpe condition and, \reviewed{consequently, in} unique real-valued parameters $\zeta_i$. 
%Introducing parameters $\lambda_i$ (\ref{eq:def_lambda}),
Solving for the ICs restricts the equilibrium ICs $x_i^*(0)$ (\ref{eq:ic_constraint}) to ensure $u^*>0$.
%The definition of $\zeta_i$ contains the same equilibrium input $u^*$ for $i \in \{1,2\}$. Equating both conditions (\ref{eq:u_star_constraint}), introducing the parameters $\lambda_i$ (\ref{eq:def_lambda}), and solving for the initial conditions restricts possible equilibrium by constrained initial conditions $x_i^*(0)$ (\ref{eq:ic_constraint}) that ensure a positive equilibrium dilution $u^*$.
\hfill$\Box$\end{pf}

\subsection{System Transformation}
\reviewed{In a next step, a system transformation introduced in \cite{karafyllis2017stability} is applied to (\ref{eq:sys}), in order to express the IPDE system terms of two coupled ODEs and two autonomous (but exponentially stable) IDEs. This serves as basis for designing a stabilizing feedback law.}
%To get there, the relationship between hyperbolic PDEs and integral delay equations (IDEs) is exploited to split off a two-dimensional ODE from the infinite-dimensional IPDEs (\ref{eq:ipde1}) and (\ref{eq:ipde2}) by a system transformation.
\begin{proposition}[System Transformation]
    Consider the mapping \reviewed{from \cite{karafyllis2017stability}}
    \begin{align}
    \begin{bmatrix}
    \eta_1(t)\\\eta_2(t)\\\psi_1(t-a)\\\psi_2(t-a)
    \end{bmatrix}=
    \begin{bmatrix}
        \mathrm{ln}\left(\Pi_1[x_{1}](t)\right)\\
        \mathrm{ln}\left(\Pi_2[x_{2}](t)\right)\\
        \frac{x_{1}(a,t)}{x^*_1(a)\Pi_1[x_{1}](t)}-1\\
        \frac{x_{2}(a,t)}{x^*_2(a)\Pi_2[x_{2}](t)}-1
    \end{bmatrix}, \label{eq:system-trafo}
\end{align}
defined with the functionals 
\begin{align}
    \Pi_i[x_i](t)
    &= \frac{\int_0^A \pi_{0,i}(a) x_i(a,t)  d a}
    {\int_0^A a k_i(a) x_i^*(a) d a}, \label{eq:pi-function}
\end{align}
where
%\begin{align}
%    \pi_{0,i}(a) =\int_a^A k_i(s) e^{{\int_s^a \zeta_i + \mu_i(l) d l}} d s
%\end{align}
$\pi_{0,i}(a) =\int_a^A k_i(s) e^{{\int_s^a \zeta_i + \mu_i(l) d l}} d s$
are the adjoint eigenfunctions to the zero eigenvalue of the adjoint differential operator \cite{schmidt2018yield} 
%\begin{align}
%    \mathcal{D}_i^* \pi_{0,i}(a) = \frac{d\pi_{0,i}(a)}{da} - (\mu_i(a) + \zeta_i) \pi_{0,i}(a) + k_i(a) \pi_{0,i}(0).
%\end{align}
\begin{multline}
    \mathcal{D}_i^* \pi_{0,i}(a) = \frac{d\pi_{0,i}(a)}{da} - (\mu_i(a) + \zeta_i) \pi_{0,i}(a) \\ + k_i(a) \pi_{0,i}(0).
\end{multline}
The transformed variables satisfy the transformed system
\begin{subequations}
\label{eq:transformed-system}
    \begin{align}
        \dot{\eta}_1(t) &= \begin{multlined}[t][0.6\columnwidth]\zeta_1
        - e^{\eta_2(t)} \int_0^A b_1(a) x_2^*(a) \\ \cdot \left(1 + \psi_2(t-a)\right) d a
        \end{multlined} \label{eq:doteta1} \\
        %\dot{\eta}_1(t) &= \lambda_2  
        %- e^{\eta_2(t)} \int_0^A b_1(a) x_2^*(a) \left(1 + \psi_2(t-a)\right) d a \label{eq:doteta1} \\
        \dot{\eta}_2(t) &= \begin{multlined}[t][0.65\columnwidth] \zeta_2  -u(t)
        - e^{\eta_1(t)} \int_0^A b_2(a) x_1^*(a) \\ \cdot \left(1 + \psi_1(t-a)\right) d a \end{multlined}  \label{eq:doteta2} \\
        %\dot{\eta}_2(t) &= \lambda_1  %\nonumber \\
        %- e^{\eta_1(t)} \int_0^A b_2(a) x_1^*(a) \left(1 + \psi_1(t-a)\right) d a + u^* - u(t) \label{eq:doteta2} \\
        \psi_i(t) 
        &= \int_0^A \tilde{k}_i(a) \psi_i(t-a) d a \label{eq:psi-def}\\
        %=& \int_0^A k_i(\alpha) e^{- \int_0^{{\alpha}} \left( \mu_i(s) + \zeta_i \right) d s} \psi_i(t-\alpha) d \alpha \label{eq:psi-def}\\
        %\psi_i(t) &= \langle \tilde{k}_i(\cdot), \psi_i(t-a) \rangle \label{eq:psi-def}\\
        \eta_i(0) &= \ln \left( \Pi [x_{i,0
        }]\right)=:\eta_{i,0}\\
        \psi_i(-a) &= \frac{x_{i,0}(a)}{x_i^*(a)\Pi [x_{i,0
        }]} - 1=: \psi_{i,0}(a)\label{eq:psi0}
    \end{align}
\end{subequations}
and, the unique solution of system (\ref{eq:sys}) is given by
\begin{align}
    x_i(a,t) = x_i^*(a)e^{\eta_i(t)}(1+\psi_i(t-a)).\label{eq:solution-x}
\end{align}
\end{proposition}

\begin{pf}
The proof follows the proof of Proposition~2 from \cite{kurth2022model} without intraspecific competition, i.e. $p_i(a)=0$.
\hfill$\Box$\end{pf}

\new{\begin{proposition}[Equilibrium of the transformed system]\label{prop:equilibrium-transformed}
    The equilibrium state $(\eta_i^*, \psi_i^*)$ of the transformed system~\eqref{eq:transformed-system}, under the equilibrium input $u^*=\zeta_2-\lambda_1$~\eqref{eq:u_star_constraint}, is $(\eta_1^*, \eta_2^*) = (0,0)$, $\rev{(\psi_1^*, \psi_2^*) = (0,0)}$. 
    % \begin{align}
    %     (\eta_1^*, \eta_2^*) = (0,0),\\
    %     \rev{(\psi_1^*, \psi_2^*) = (0,0)}.
    % \end{align}
\end{proposition}
\rev{
\begin{pf}
    Substituting the equilibrium population profiles~\eqref{eq:ss_profiles} into the transformation functionals \eqref{eq:pi-function} results in
    \begin{equation}
        \Pi_i[x_i^*] = \frac{\int_0^A \int_a^A k_i(s) e^{-\int_0^s (\zeta_i + \mu_i(\ell)) d\ell} ds\ da}{\int_0^A a k_i(a) e^{-\int_0^a (\zeta_i + \mu_i(\ell)) d\ell} da}  
    \end{equation}
    where the denominator is result of applying the Cauchy formula for repeated integration to simplify the numerator. Thus, $\Pi_i[x_i^*] =1$ at the equilibrium and, consequently, $\eta_i^* = 0$ for the equilibrium population distribution. With $\Pi_i[x_i^*]=1$ and $x_i(a,t) = x_i^*(a)$, $\psi_i(t) = 0$ for all $t\in[-A,\infty)$. This is verified with \eqref{eq:doteta1}-\eqref{eq:psi-def}. Equation \eqref{eq:psi-def} is satisfied for all $t\in[-A,\infty)$ with $\psi_i^*(t)=\psi_i^*=0$. Additionally, \eqref{eq:doteta1}-\eqref{eq:doteta2} with $\dot{\eta}_1=\dot{\eta}_2 = 0$ results in 
    \begin{align}
        \eta_1^* &= \ln \left( \frac{\zeta_2-u^*}{\int_0^A b_2(a)x_1^*(a)(1+\rev{\psi_{1}^*})da} \right), \\
        \eta_2^* &= \ln \left( \frac{\zeta_1}{\int_0^A b_1(a)x_2^*(a)(1+\rev{\psi_{2}^*})da} \right).
    \end{align}
    Making use of the definition of $\lambda_i$ in \eqref{eq:def_lambda} and the resulting constraints on $\zeta_i$, $u^*$, we conclude that the ODE's equilibrium is  $(\eta_1^*, \eta_2^*) = (0,0)$.
\hfill$\Box$\end{pf}
}}
\rev{Instead of dealing with an IPDE system, the population dynamics are now expressed in terms of an actuated ODE system ($\eta_i$) and autonomous IDEs ($\psi_i$), referred to as internal dynamics. This is particularly convenient because} it was proved in \cite{karafyllis2017stability} that the states $\psi_i$ 
% of these internal dynamics 
are restricted to the sets
\begin{align}
    \mathcal{S}_i= \Big\{ &\psi_i \in C^0([-A,0];(-1,\infty)): \nonumber \\
    &P(\psi_i)=0 \land \psi_i(0) = \int_0^A \tilde{k}_i(a) \psi_i(-a) d a \Big\}
\end{align}
where
\begin{align}
    P(\psi_i) = \frac{\int_0^A \psi_i(-a) \int_{a}^A \tilde{k}_i(s) d s d a}{ \int_0^A a \tilde{k}_i(a) d a },
\end{align}
and that they are globally exponentially stable in the $\mathcal{L}^{\infty}$ norm. This means that there exist $M_i\geq1$, $\sigma_i\geq0$ such that
\begin{align}
\label{eq:psi-decays}
    |\psi_i(t-a)| \leq M_i e^{-\sigma_i t}||\psi_{i,0}||_{\infty} 
\end{align}
holds for all $t\geq 0$ and $\psi_{i,0}\in \mathcal{S}_i$ and all $a\in[-A,0]$.

Please note that in the following the argument $(t)$ in the state $\eta_i(t)$, $\psi_i(t)$ is dropped, and $\psi_{i,t}:=\psi_i(t-a)$ denotes the ``age-history'' of $\psi_i$ at certain $t\geq0$.
Furthermore, the vector states
$\eta:=[\eta_1, \eta_2]$, 
$\psi:=[\psi_1, \psi_2]$, 
are introduced for a more concise notation.

\section{Control Design}\label{sec:control-design}
\rev{
Transforming the IPDE system~\eqref{eq:sys} into an (unstable) actuated ODE subsystem \eqref{eq:doteta1}, \eqref{eq:doteta2} coupled with an autonomous but exponentially stable IDE subsystem \eqref{eq:psi-def} (referred to as the internal dynamics) allows for initially neglecting the stable IDE component. }
Hence, for the sake of control design, the ODE system with $\psi_i \equiv 0$, 
\new{\begin{subequations}\label{eq:ode-system-before-phi}
    \begin{align}
        % \dot{\eta}_1 &= \lambda_2(1-e^{\eta_2}),    \\
        % \dot{\eta}_2 &=  \lambda_1(1-e^{\eta_1}) +  u^* - u,
        \dot{\eta}_1 &= \zeta_1 - \lambda_2 e^{\eta_2},    \\
        \dot{\eta}_2 &=  \zeta_2 - \lambda_1 e^{\eta_1} +  u^* - u,
    \end{align}
\end{subequations}
is considered.} Using the definition of $\zeta_i$~\eqref{eq:ss} and $\lambda_i$~\eqref{eq:def_lambda}, (\ref{eq:ode-system-before-phi}) is rewritten as
\begin{subequations}\label{eq:ode-system}
    \begin{align}
        \dot{\eta}_1 &= \phi_2(\eta_2),    \\
        \dot{\eta}_2 &=  \phi_1(\eta_1) +  u^* - u .
    \end{align}
\end{subequations}
with functions
\begin{align}\label{eq:phi-def}
    \phi_i(\eta_i) &:= \lambda_i\left(1 - e^{\eta_i}\right).
    % \phi_2(\eta_2)&:= \lambda_2 \left(1-  e^{\eta_2} \right),
\end{align}
Both $\phi_1$ and $\phi_2$ are decreasing in the concentration of the respective other predator $\eta_2$, $\eta_1$,
\new{which have a repressive effect on one another. The dilution $u$ only has a repressive effect on the harvested predator. }
Hence, the unharvested predator and the dilution work in tandem relative to the harvested predator population, posing a risk of overharvesting. 
\new{This compels a choice of a feedback law that may take negative values, to compensate for the fact that, when the harvested predator population is depleted, positive dilution may result in the extinction of both populations.}

\new{\textbf{Open-Loop Instability.}} 
\new{In the uncontrolled case $u=u^*$,~\eqref{eq:ode-system} becomes
\begin{subequations}\label{eq:ode-system-uncontrolled}
    \begin{align}
        \dot{\eta}_1 &= \phi_2(\eta_2),    \\
        \dot{\eta}_2 &=  \phi_1(\eta_1).
    \end{align}
\end{subequations}
The unique equilibrium of~\eqref{eq:ode-system-uncontrolled} is $(\eta_1, \eta_2)=(0,0)$
and the Jacobian 
\new{is $\begin{bmatrix} 0 & -\lambda_2\\ -\lambda_1 & 0 \end{bmatrix}$,}
with saddle-point eigenvalues $s_{1,2} = \pm \sqrt{\lambda_1 \lambda_2}=\pm \sqrt{\zeta_1 (\zeta_2-u^*)}$.} 

\smallskip
\textbf{A Backstepping Feedback Design.}
For the control design, the backstepping transformation 
% $z = \eta_2 - c_1\eta_1, \ c_1>0$
\begin{align}
    z = \eta_2 - c_1\eta_1, \ c_1>0
\end{align}
is introduced. Straightforward calculation then yields
\begin{subequations}
    \begin{align}
        \dot \eta_1 &= \lambda_2(1-e^{c_1\eta_1}) - \lambda_2 e^{c_1\eta_1}(e^z-1) \\
        \dot \eta_2 &= \lambda_1(1-e^{\eta_1}) + u^* - u \\
        \dot z &= \lambda_1(1-e^{\eta_1}) - c_1 \lambda_2 (1-e^{\eta_2}) + u^*- u.
    \end{align}
\end{subequations}
Before designing a controller, the two positive definite and radially unbounded functions,
\begin{align}
\label{eq-omega-mu}
    \omega(q) = e^q - 1 - q ,\quad
    \mu(q) = \sinh^2\left(\frac{q}{2}\right),
\end{align}
are introduced.
Next, consider the Lyapunov functions
\begin{align}
    V_1(\eta_1) &= \omega(-c_1 \eta_1),\label{eq:V1}\\
    V_2(\eta_1, \eta_2) &= \omega(\eta_2-c_1\eta_1)=\omega(z) \label{eq:V2}
\end{align}
and the overall Lyapunov function
\begin{align}
    V_3(\eta_1, \eta_2) &= \theta V_1(\eta_1) + \reviewed{V_2(\eta_1,\eta_2)} \notag\\
    &= \theta \omega(-c_1 \eta_1) 
    +\omega(\eta_2-c_1\eta_1), \ \theta>0.\label{eq:lyapunov-ode}
\end{align}
Noting that $\omega'(q)=e^q-1$, as well as that
\begin{align}
    (e^q-1)(e^{-q}-1) = -4 \sinh^2\left(\frac{q}{2}\right),\label{eq:sinh-property}
\end{align}
\reviewed{the derivatives of the Lyapunov functions $V_1$, $V_2$ result in}
\new{\begin{align}
    \dot V_1 
    =& \ c_1 \dot \eta_1 (1-e^{-c_1 \eta_1}) \notag \\
    =& \ c_1 \lambda_2 (1-e^{c_1\eta_1})(1-e^{-c_1 \eta_1}) \notag\\
    &-  c_1 \lambda_2 e^{c_1 \eta_1} (e^{z}-1)(1-e^{-c_1\eta_1})\notag  \\
    =& -4c_1 \lambda_2 \mu(-c_1 \eta_1) - c_1 \lambda_2 (e^z-1)(e^{c_1 \eta_1} - 1),\\
    \dot V_2 
    =& \ \dot z (e^z - 1) \notag \\
    =& \ (e^z - 1) \left( \lambda_1(1-e^{\eta_1}) - c_1\lambda_2(1-e^{\eta_2}) + u^* - u \right).
\end{align}}
\mycomment{Constructing $\dot V_3$ based on \eqref{eq:lyapunov-ode}, we note that a negative definite Lyapunov derivative aligning with our backstepping transformation and the construction of $V_3$ is
\begin{align}
    \dot V_3 = - 4 \lambda_2 \left[ \theta c_1 \mu(-c_1 \eta_1) + c_2 \mu(\eta_2 - c_1\eta_1) \right],
\end{align}
which is achieved through the feedback law
\begin{align}
    u = u^* + \lambda_2 \Big[ - c_2(e^{-z} - 1) - \theta c_1 (e^{c_1 \eta_1} -1) \notag \\
    + \frac{\lambda_1}{\lambda_2} (1-e^{\eta_1}) - c_1  (1 - e^{\eta_2}) \Big]\label{eq:control-law}
\end{align}
with $c_2>0$}
The feedback law
\begin{align}
    u = u^* + \lambda_2 \Big[ - c_2(e^{-z} - 1) - \theta c_1 (e^{c_1 \eta_1} -1) \notag \\
    + \frac{\lambda_1}{\lambda_2} (1-e^{\eta_1}) - c_1  (1 - e^{\eta_2}) \Big]\label{eq:control-law}
\end{align}
with $c_2>0$ then produces the overall Lyapunov derivative
\begin{align}
    \dot V_3 = - 4 \lambda_2 \left[ \theta c_1 \mu(-c_1 \eta_1) + c_2 \mu(\eta_2 - c_1\eta_1) \right],
\end{align}
which is negative definite \reviewed{and aligns with our backstepping transformation and construction of $V_3$.}
Global asymptotic stability of the equilibrium $\eta=0$ then follows.

\new{\begin{theorem} \label{thm:global-stabilization-ode}
Under the feedback law \eqref{eq:control-law}, the equilibrium $\eta=0$ of the system \eqref{eq:ode-system} is globally asymptotically and locally exponentially stable. The control signal $u(t)$ remains bounded though not necessarily positive. Furthermore, $u(t)>0$ for all $t\geq 0$ and
for all $\eta(0)$ belonging to the largest level set of $V_3(\eta)$  within the set
% \begin{align}
%     \mathcal{D}_0 = \Big\{ \eta \in \mathbb{R}^2  \Big| &     u^* + \lambda_2 \Big[ - c_2(e^{-\eta_2+c_1\eta_1} - 1) \notag \\
%     & - \theta c_1 (e^{c_1 \eta_1} -1) + \frac{\lambda_1}{\lambda_2} (1-e^{\eta_1}) 
%     \notag \\
%     &- c_1  (1 - e^{\eta_2}) \Big] > 0 \Big\}.
% \end{align}
\rev{
\begin{align}
    \mathcal{D}_0 = \Big\{ x\in\mathcal{F} \Big| &     u^* + \lambda_2 \Big[ - c_2 \left (\frac{(\Pi_1[x_1])^{c_1}}{\Pi_2[x_2]} - 1 \right) \notag \\
    & - \theta c_1 \left((\Pi_1[x_1])^{c_1} -1\right) + \frac{\lambda_1}{\lambda_2} (1-\Pi_1[x_1]) 
    \notag \\
    &- c_1  (1 - \Pi_2[x_2]) \Big] > 0 \Big\}.
\end{align}}
\end{theorem}
}

\section{Stability With Nonzero $\psi$}\label{sec:stability-results}
\rev{Up to now, the internal dynamics $\psi_i$ were neglected in order to facilitate control design. Now,} it is shown that the control law (\ref{eq:control-law}) based on the ODE system neglecting the internal dynamics (\ref{eq:ode-system}) \rev{also} stabilizes the full ODE-IDE system (\ref{eq:transformed-system}).
For this, the first step is to introduce the mapping $v_i:\mathcal{S} \rightarrow \mathbb{R}_+$,
\begin{align}
    v_i(\psi_{i,t}) = \ln \left(1 + \int_0^A \bar{b}_j(a) \psi_i(t-a) d a\right),
    \label{eq:definition-v-map} %
\end{align}
for $i,j \in \{1,2\}$, $i\neq j$, with 
\begin{align}
    \bar{b}_i(a) & = \frac{b_i(a)\xstar{j}}{\int_0^A b_i(\alpha)x^\ast_j(\alpha) d \alpha}, \ \int_0^A \bar{b}_i(a) d a = 1.\label{eq:definition-b-bar}
\end{align}
Then, (\ref{eq:doteta1}), (\ref{eq:doteta2}) are rewritten as
\begin{subequations}
    \begin{align}
        \dot{\eta}_1 &= 
        \phi_2(\eta_2+v_2(\psi_2)),    \\
        \dot{\eta}_2 &=  
        \phi_1(\eta_1+v_1(\psi_1)) +  u^* - u .
    \end{align}
\end{subequations}
For better readability in the following, denote
\begin{alignat}{2}
\label{psihat-defs}
    \phi_i &:= \phi_i(\eta_i), \ \ \ \hat{\phi}_i &:= \phi_i(\eta_i+v_i).
\end{alignat}
Additionally, a technical assumption on the birth kernel is needed for the definition of a Lyapunov function \cite{karafyllis2017stability}.

\begin{assumption}[Birth kernel \cite{karafyllis2017stability}]\label{assump:technical_assumption}
There exist constants $\kappa_i>0$ such that
$\int_0^A |\tilde{k}_i(a)- z_i \kappa_i \int_a^A \tilde{k}_i(s) d s | d a < 1$ with $\reviewed{z_i = (\int_0^A a \tilde{k}_i(a)  d a)^{-1}}$. 
\end{assumption}

The technical Assumption~\ref{assump:technical_assumption} is mild, satisfied by arbitrary mortality rate $\mu_i$ for every birth kernel $k_i$ that has a finite number of zeros on $[0, A]$, \rev{imposing no significant restriction to the class of systems for which the following stability results apply, and especially holds true for bioreactor populations.} 
The role of Assumption~\ref{assump:technical_assumption} is crucial for the establishment of the function $G$ \new{that is necessary to construct Lyapunov functions}. Means of verifying the validity of Assumption~1 and detailed discussions are given in \cite{karafyllis2017stability}.
If Assumption 1 holds then there exist constants $\sigma_i>0$, with $i=1,2$, such that the inequalities $\int_0^A | \tilde{k}_i - z_i \kappa_i \int_a^A \tilde{k}_i(s) d s | e^{\sigma_i a} d a <1 $ for $i=1,2$ hold.

With this, the functionals $G_i$ defined as
\begin{align}
        G_i(\psi_{i}) &: = \frac{\max_{a \in [0,A]} |\psi_i(-a)| e^{\sigma_i ({ A} - a)}}{1 + \min(0, \min_{a \in [0,A]} \psi_i(-a))} 
        , \label{eq:g-functional}
\end{align}
are introduced. These functionals are positive definite in $\psi_i$ in the sense of the maximum norm in $a$, and, recalling from \cite{karafyllis2017stability}, their Dini derivatives ($D^+$) satisfy 
\begin{align}
D^+ G_i(\psi_{i,t}) \leq -\sigma_i G_i(\psi_{i,t}) 
\label{eq:dini_g}
\end{align}
along solutions $\psi_{i,t}$ of the IDE for with sufficiently small parameters $\sigma_i>0$. This property follows from Corollary~4.6 and the proof of Lemma~5.1 of \cite{karafyllis2017stability} with 
$C_i(\psi_{i})=\frac{1}{(1 + \min(0, \min_{a \in [0,A]} \psi_i(-a)))^2}$ and $
b(s) = s $, similar to what is discussed in \cite{haacker2024stabilization}. 

Further, to state the main stability result, the following lemmas are needed \rev{for mathematical purposes.}
\begin{lemma}\label{lem:b-value}
    Consider the parametrized function 
    \begin{align}
        f(y;r,\beta) = \frac{\beta r \omega'(ry)}{1 +\beta \omega(ry)} = \frac{r(e^{ry}-1)}{e^{ry}-ry-1+\beta^{-1}} \label{eq:f-function}
    \end{align}
    with $r\neq 0$ and $\beta>0$. Then,
    \reviewed{\begin{eqnarray}
    \max_{y \in \mathbb{R}} \left\vert f(y;r,\beta)\right\vert &=& 
    \vert r\vert   B(\beta)
    \\
    B(\beta) & :=&  \frac{e^{1+\beta^{-1} + d(\beta)} -1}{e^{1+\beta^{-1} + d(\beta)} -d(\beta) - 2}\notag
    %\\ &=:  \vert r \vert
    % \vert r\vert  \frac{e^{1+\beta^{-1} + d(\beta)} -1}{e^{1+\beta^{-1} + d(\beta)} -d(\beta) - 2}\notag\\ &=:  \vert r \vert B(\beta) 
    ,\label{eq:b-value}
\end{eqnarray}}%
where $d(\beta) = W_0\left(-e^{-1 - \beta^{-1}}\right)$ and $W_0$ denotes the principal branch of the Lambert $W$ function.
\end{lemma}
\begin{pf}%[Proof of Lemma \ref{lem:b-value}.]
    The denominator of $f(y;r,\beta)$ is $>0$, so $f$ has no singularities, and both limits exist, i.e. $\lim_{ry\to-\infty}f=0$, $\lim_{ry\to\infty}f=r$. Therefore, $\vert f \vert$ is maximized at the extremum of $f$.
%To find $B$, let us find the extrema $y^*(r,\theta)$ of $f$. 
By the first order necessary conditions for extrema, $y^*$ is located at points where either $\frac{\partial f}{\partial y}$ is not defined or when $\frac{\partial f}{\partial y} = 0$. Thus, the extrema $y^*(r,\beta)$  of $f$ satisfy
% \new{\begin{align}
%    r^2 e^{ry^*(r,\beta)} \left(e^{ry^*(r,\beta)} - ry^*(r,\beta) - 1 + \beta^{-1}\right)\notag \\
%    - r^2 \left(e^{ry^*(r,\beta)} - 1\right)^2 = 0,
% \end{align}
% or more simply}
\begin{equation}
    \left(ry^*(r,\beta) - 1 - \beta^{-1}\right)e^{ry^*(r,\beta)} = -1.
\end{equation}
Using the principal branch of the Lambert $W$ function, it is possible to get an explicit form
\begin{equation}
    y^*(r,\beta) = r^{-1} \left[1 + \beta^{-1} + W_0\left(-e^{-1-\beta^{-1}}\right)\right]
\end{equation}
% Thus, the maximum magnitude of $f$ is explicitly written as $\vert r\vert B(\beta)$~\eqref{eq:b-value}.
for the global maximizer $y^*$ of $\vert f \vert$. The global maximum of $\vert f \vert$ is $\vert r\vert B(\beta)$ with $B(\beta)>0$ for any $\beta > 0$ (cf. Figure~\ref{fig:b-curve}).
\mbox{}\hfill$\Box$\end{pf}

\begin{figure}[t]
    \centering
    \includegraphics{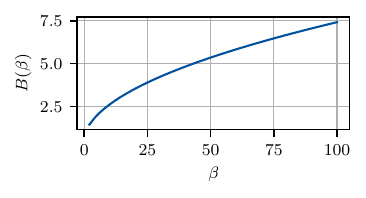}
    \caption{The function $B(\beta)$ is a monotonically increasing positive function defined for $\beta\in(0,\infty)$ and with an image of $(0,\infty)$.}
    \label{fig:b-curve}
\end{figure}

\new{\begin{lemma}\label{lm:h-function}
The function
\begin{align}
\label{h-Lyapunov}
       h(p)&:= %\frac{1}{\sigma} 
       \int_0^p \frac{1}{z} (e^z-1) d z.
\end{align}
is positive definite and radially unbounded.
\end{lemma}
\begin{pf}
    For $z\geq0$, by Taylor expansion,  
$\frac{{\rm e}^z-1}{z} = 1+ \frac{z}{2} +r(z) \geq 1+ \frac{z}{2}$, where $r(0)=0$ and $r(z)\geq 0$  contains powers of $z$ of second and higher orders. Hence, for $p\geq 0$, it holds that $h(p) =  p+\frac{p^2}{2} +\int_0^p r(s)ds \geq p+\frac{p^2}{2}$, which is radially unbounded in $p\geq 0$.
\hfill$\Box$\end{pf}
}

\new{Lemma~\ref{lem:b-value} provides an estimate to upper bound terms occuring in the Lyapunov derivative, whereas Lemma~\ref{lm:h-function} is needed for constructing the Lyapunov function itself. With these prerequisites, the main result can now be stated.}
\smallskip

\begin{theorem}\label{thm:local-stability}
Let Assumption~\ref{assump:technical_assumption} hold and define $\mathcal{S}=\mathcal{S}_1 \times \mathcal{S}_2$. Consider system (\ref{eq:transformed-system}) with the control law (\ref{eq:control-law}),
in original coordinates given as
\begin{align}
\label{eq-controller-original-variables}
    u =& u^* + \lambda_2 \Big[ c_2 \left(1-\frac{(\Pi_1(x_1))^{c_1}}{\Pi_2(x_2)}\right)\notag \\
   &+\theta c_1 \left( 1-(\Pi_1(x_1))^{c_1}  \right) + \frac{\lambda_1}{\lambda_2} \left(1- \Pi_1(x_1)\right) \notag \\
   &+c_1(\Pi_2(x_2)-1)  \Big],
\end{align} %
on the state space $\mathbb{R}^2 \times S$, which is a subset of the Banach space $\mathbb{R}^2 \times C^0 ([-A,0] ; \mathbb{R}^2)$ with the standard topology.  
Suppose the parameters satisfy $c_1>0$, $c_2>0$, $\theta>0$.
Denote $\eta=[\eta_1, \eta_2]^\top$, $\psi = [\psi_{1}, \psi_{2}]^\top$ and the Lyapunov functional 
\begin{align}
        V(\eta, \psi) &= \ln(1+V_3(\eta))
        + \frac{\gamma_1 }{\sigma_1}
        h(G_1(\psi_{1})) 
        + \frac{\gamma_2 }{\sigma_2}
        %e^{\sigma_2 A}} 
        %\gamma_2 
        h(G_2(\psi_{2})) 
        \label{eq:lyapunov-ode-ide}
\end{align}
with $V_3(\eta)$ from (\ref{eq:lyapunov-ode}), the positive weights chosen as
\begin{subequations}
    \begin{align}\label{eq:gamma-restrictions}
        \gamma_1 &> 2\lambda_1 B(1),\\
        \gamma_2 &> 2\lambda_2 c_1(B(1) + B(\theta)) ,
    \end{align}%
\end{subequations}%
and the positive definite radially unbounded function \eqref{h-Lyapunov}.
% %$h(\cdot)$ defined as
% \begin{align}
% \label{h-Lyapunov2}
%        h(p)&:= %\frac{1}{\sigma} 
%        \int_0^p \frac{1}{z} (e^z-1) d z.
% \end{align}
Then, 
\begin{enumerate}
    \item positive invariance holds for all the level sets of $V$ of the form
\begin{align}
\label{eq:Omegac}
    \Omega_c:=\{\eta \in \mathbb{R}^2, \psi \in \mathcal{S} \ |  \ V(\eta, \psi)\leq c\}
\end{align} 
that are within the set 
% \begin{align}
% \label{eq:constraints}
%     \mathcal{D} := &\Biggl\{ \eta \in \mathbb{R}^2, \psi \in \mathcal{S} \ 
%        \Bigg| \nonumber\\
%        & \ \eta_1\leq\ln\left(\frac{\gamma_1}{2\lambda_1 B(1)}\right), 
%       \notag \\
%     & \ \eta_2\leq \ln\left(\frac{\gamma_2}{2\lambda_2 c_1(B(1) + B(\theta))}\right), 
%     \nonumber\\
%     &  \ \ 
%     u^* + \lambda_2 \Big[ - c_2(e^{-z} - 1) - \theta c_1 (e^{c_1 \eta_1} -1) \notag \\
%     & \quad + \frac{\lambda_1}{\lambda_2} (1-e^{\eta_1}) - c_1  (1 - e^{\eta_2}) \Big] >
%  0 \Biggr\},
% \end{align}
\rev{
\begin{align}
\label{eq:constraints}
    \mathcal{D} := &\Biggl\{ x \in \mathcal{F} \ 
       \Bigg|\ \Pi_1[x_1] \leq\frac{\gamma_1}{2\lambda_1 B(1)}, 
      \notag \\
    & \ \Pi_2[x_2]\leq \frac{\gamma_2}{2\lambda_2 c_1(B(1) + B(\theta))}, 
    \nonumber\\
    &  \ \ 
    u^* + \lambda_2 \Big[ - c_2(\frac{(\Pi_1[x_1])^{c_1}}{\Pi_2[x_2]} - 1) \notag \\
    & \quad - \theta c_1 ((\Pi_1[x_1])^{c_1} -1) 
     + \frac{\lambda_1}{\lambda_2} (1-\Pi_1[x_1])\notag \\
    & \quad- c_1  (1 - \Pi_2[x_2]) \Big] >
 0 \Biggr\},
\end{align}
}
namely, within the set $\mathcal{D}$ where none of the predator's surplus is too big;
\item the input $u(t)$ remains positive for all time; 
\item there exists $\theta_0\in\mathcal{KL}$ such that, for all initial conditions $(\eta_0, \psi_0)$ within the largest $\Omega_c$ contained in $\mathcal{D}$, the following estimate holds, 
\begin{equation}
\label{eq:KL-estimate}
    \left|(\eta(t),G(t)) \right| \leq \theta_0\left(\left|(\eta(0),G(0) \right|, t\right), \quad\forall t\geq 0\,,
\end{equation}
where $G(t) = \left(G_1(\psi_{1,t}),G_2(\psi_{2,t}) \right)$;
\item \new{the equilibrium $\eta=0, \psi=0$ is locally exponentially stable in the norm $\sqrt{\eta_1^2+\eta_2^2} + \|\psi_1\|_\infty+ \|\psi_2\|_\infty$ 
%, and no eigenvalue of the system linearized at the origin has real part larger than $\max\{-\sigma_1, -\sigma_2, -\lambda_1, -\lambda_2(1+c_1+c_2)\}$.
    with an exponential decay rate of at least $\min\left\{\frac{\sigma_1}{1+\varepsilon},\frac{\sigma_2}{1+\varepsilon}, -\frac{\lambda_2}{1+\varepsilon} {\rm Re}(p_1),-\frac{\lambda_2}{1+\varepsilon}{\rm Re}(p_2)\right\}>0$, where $\sigma_1,\sigma_2$ are defined in \eqref{eq:psi-decays}, $p_1,p_2$ are the roots of the polynomial $s^2 + (c_1+c_2) s + \theta + c_1c_2$, and $\varepsilon>0$ is arbitrarily small. }
\end{enumerate}
\end{theorem}

\medskip

To illustrate this complex interconnection of multiple constraints of $\mathcal{D}$, Figure~\ref{fig:roas} shows level sets of $V_3(\eta)$ in the $\eta$-plane, namely, in the $\Pi_1(x_1)$-$\Pi_2(x_2)$-plane, together with the domain $\mathcal{D}$ in gray, i.e., the constraint $u>0$ for the ODE-system (\ref{eq:ode-system}) that also fulfills $\Pi_1(x_1)<\frac{\gamma_1}{2\lambda_1 B(1)}$ and $\Pi_2(x_2)<\frac{\gamma_2}{2\lambda_2 ( c_1 B(1) + c_1B(\theta))}$. Note that the set $\mathcal{D}$ allow arbitrary functions $(\psi_1,\psi_2)$ in the function space $\mathcal{S}$. This means that the set $\mathcal{D}$ is infinite dimensional and our planar plot in Figure~\ref{fig:roas} is only the $\psi_1\equiv \psi_2 \equiv 0 $ slice of $\mathcal{D}$. 

\smallskip

\begin{figure}[t]
    \centering
    \includegraphics{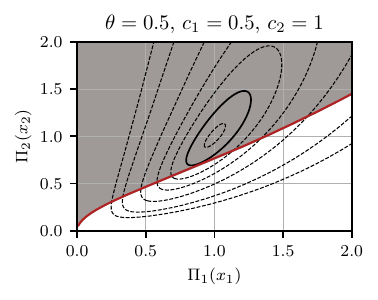}
    \caption{Level sets of $V_3$ (\mlLineLegendDashed{black}).
    The gray area is the set $\mathcal{D}$, namely, the set in which  $\Pi_1[x_1]\leq\ln\left(\frac{\gamma_1}{2\lambda_1 B(1)}\right)$, $\Pi_2[x_2]\leq \ln\left(\frac{\gamma_2}{2\lambda_2 (c_1 B(1) + c_1B(\theta))}\right)$, and $u>0$. The red curve corresponds to $u=0$.
    The largest level set of $V_3$ (\mlLineLegend{black}) is contained within the actual region of attraction of the origin.%, for the case $\psi=0$. 
    }
    \label{fig:roas}
\end{figure}

\begin{pf}[Proof of Theorem~\ref{thm:local-stability}]
    Note that the existence and uniqueness of the solution $\psi_{i,t} \in \mathcal{S}_i$ and the property
    \begin{align}
        \inf_{t\geq A} \psi_i(t) \geq \min_{t \in [-A, 0]} \psi_i(t) > -1 \ \forall t \geq 0\label{eq:psi_lower_bound}
    \end{align}
    is provided in Lemma~4.1 of \cite{karafyllis2017stability}. Further, the IC (\ref{eq:psi0}) of the IDEs is lower bounded as $\psi_{i,0}(a)>-1$ by definition.
    With $\bar{b}_i(a) \geq 0$, $\int_0^A \bar{b}_i(a)d a=1$ by definition (\ref{eq:definition-b-bar}), the map $v_i(\psi_{i,t})$ (\ref{eq:definition-v-map}) is well-defined and continuous. Hence, the ODE subsystem of 
    %(\ref{eq:ode-ide-closed-loop}) 
    the closed loop, i.e. of system~\eqref{eq:transformed-system} with~\eqref{eq:control-law}, locally admits a unique solution.
    From \cite{karafyllis2017stability}, more precisely (A.43), it follows that 
    \begin{align}
    |v_i(\psi_{i,t})| \leq G_i(\psi_{i,t})\label{eq:v_leq_G},
    \end{align}
    where $v_i:\mathcal{S}\rightarrow \mathbb{R}_+$ is defined in (\ref{eq:definition-v-map}), holds. \new{For improved clarity, the proof of Theorem~\ref{thm:local-stability} is split into several parts.}
    
    \textbf{Claims 1 and 2: Lyapunov level sets.}
    Let's recall that the control law contains the functions $\phi_i$ and rewrite~\eqref{eq:control-law} for better readability as
    \begin{align}
    u
    =& u^* + \lambda_2 \left[ - c_2(e^{-z} - 1) - \theta c_1 (e^{c_1 \eta_1} -1)\right]\notag \\
    &+ \phi_1(\eta_1) - c_1  \phi_2(\eta_2)\reviewed{.} 
    \end{align}
    Following the same backstepping approach as in Section~\ref{sec:control-design} with $z = \eta_2 - c_1\eta_1, \ c_1>0$, the derivative of $V_3(\eta)$ is
    \begin{align}
         \dot V_3 
         =& -4\theta c_1 \lambda_2 \mu(-c_1 \eta_1) +\reviewed{c_2}\lambda_2 (e^z-1)(e^{-z}-1) \notag \\
        &+ (e^z-1)\big[ (\hat{\phi}_1 - \phi_1) \notag \\
        &- c_1(\hat{\phi}_2 - \phi_2)\big] - \theta c_1(e^{-c_1\eta_1}-1)(\hat{\phi}_2 - \phi_2)\reviewed{,}
    \end{align}
    whose first row corresponds to $\dot V_3$ of the ODE system~\eqref{eq:lyapunov-ode}, and which contains additional terms exponential in $\eta_2-c_1\eta_1$ and $-c_1\eta_1$.
    Note that these terms are simply $\frac{d\omega(ax)}{dx}=a \omega'(ax)$, leading to the next Lyapunov function
    \begin{align}
        V_4(\eta_1, z) = \ln(1 + V_3(\eta_1, z)) \label{eq:v4-def}.
    \end{align}
%    \mytodo{Show that this is positive definite!! }
    The derivative of~\eqref{eq:v4-def} \new{is} %can be bounded to
    \begin{align}
    \dot{V}_4 =& 
    \frac{-4\theta c_1 \lambda_2 \mu(-c_1 \eta_1) - 4c_2 \lambda_2 \mu(z)}
    {1 + \theta \omega(-c_1 \eta_1) + \omega(z)} \notag \\
    +& \frac{\omega'(z)}{1 + \theta \omega(-c_1 \eta_1) + \omega(z)}
    \left[ (\hat{\phi}_1 - \phi_1) - c_1 (\hat{\phi}_2 - \phi_2) \right] \notag\\
    &- \frac{-c_1 \theta \omega'(-c_1 \eta_1)}
    {1 + \theta \omega(-c_1 \eta_1) + \omega(z)} (\hat{\phi}_2 - \phi_2) \label{eq:dot-v4-0}.% \notag \\
    % \leq& \frac{-4\theta c_1 \lambda_2 \mu(-c_1 \eta_1) - 4c_2 \lambda_2 \mu(z)}
    % {1 + \theta \omega(-c_1 \eta_1) + \omega(z)}
    %  \notag \\
    % &+ \left|\frac{\omega'(z)}{1 + \omega(z)} \right| \left|\hat{\phi}_1 - \phi_1\right|  \notag\\
    % &+ \left( c_1 \left| \frac{\omega'(z)}{1 + \omega(z)}\right|
    % + \left| \frac{-c_1 \theta \omega'(-c_1 \eta_1)}
    % {1 + \theta \omega(-c_1 \eta_1)} \right| \right) \left|\hat{\phi}_2 - \phi_2\right|,\label{eq:dot-v4-1}
    \end{align}
    \new{By recognizing} that the parametrized functions $f(z;1,1)$ and $f(\eta_1;-c_1,\theta)$ from~\eqref{eq:f-function} \new{in Lemma~\ref{lem:b-value}}
    appear within~\eqref{eq:dot-v4-0}, the derivative of $\dot V_4$ is rewritten and bounded to 
    \begin{align}
    \dot{V}_4 
    \leq& \frac{-4\theta c_1 \lambda_2 \mu(-c_1 \eta_1) - 4c_2 \lambda_2 \mu(z)}
    {1 + \theta \omega(-c_1 \eta_1) + \omega(z)}
     \notag \\
    &+ \left( c_1 \left| \frac{\omega'(z)}{1 + \omega(z)}\right|
    + \left| \frac{-c_1 \theta \omega'(-c_1 \eta_1)}
    {1 + \theta \omega(-c_1 \eta_1)} \right| \right) \left|\hat{\phi}_2 - \phi_2\right| \notag \\
    &+ \left|\frac{\omega'(z)}{1 + \omega(z)} \right| \left|\hat{\phi}_1 - \phi_1\right|.\label{eq:dot-v4-1}
    \end{align}
    % \new{by recognizing} that the parametrized functions $f(z;1,1)$ and $f(\eta_1;-c_1,\theta)$ from~\eqref{eq:f-function} \new{in Lemma~\ref{lem:b-value}}
    % appear within the derivative of $\dot V_4$. 
    According to Lemma~\ref{lem:b-value}, there exist positive functions $B(1)$,  $B(\theta)$ such that
    \new{\begin{align}% \begin{alignat}{2}
    \max_{z \in \mathbb{R}} \left|\frac{\omega'(z)}{1 + \omega(z)} \right| &= B(1) \\ % \max_{z \in \mathbb{R}} \vert f(z;1,1)\vert &&= B(1) , \\
    \max_{\eta_1 \in \mathbb{R}} \left| \frac{-c_1 \theta \omega'(-c_1 \eta_1)}
    {1 + \theta \omega(-c_1 \eta_1)} \right| &=  c_1 B(\theta) % \max_{\eta_1 \in \mathbb{R}} \vert f(\eta_1; c_1 ,\theta)\vert &&= c_1 B(\theta).
    \end{align}}
    \new{Hence, considering that $\hat{\phi}_i - \phi_i = (\phi_i - \lambda_i) (e^{v_i}-1)$, 
    \begin{align}
        \left|\frac{\omega'(z)}{1 + \omega(z)} \right| \left|\hat{\phi}_1 - \phi_1\right|  \leq  B(1)|\phi_1 - \lambda_1| |e^{v_1}-1| \label{eq:intermediate-bound-1}, 
    \end{align}
    and
    \begin{align}
         \left( c_1 \left| \frac{\omega'(z)}{1 + \omega(z)}\right|
        + \left| \frac{-c_1 \theta \omega'(-c_1 \eta_1)}
        {1 + \theta \omega(-c_1 \eta_1)} \right| \right) \left|\hat{\phi}_2 - \phi_2\right| \notag\\
        \leq  \left(c_1 B(1) + c_1B(\theta)|\right)|\phi_2 - \lambda_2| |e^{v_2}-1| \label{eq:intermediate-bound-2}
    \end{align}
    hold, and subsequently $\dot V_4$ can be bounded to
    \begin{align}
        \dot V_4 \leq& \frac{-4\theta c_1 \lambda_2 \mu(-c_1 \eta_1) - 4c_2 \lambda_2 \mu(z)}{1 + \theta \omega(-c_1 \eta_1) + \omega(z)}\notag \\
    &+ B(1)|\phi_1 - \lambda_1| |e^{v_1}-1|\notag \\
    &+ \left(c_1 B(1) + c_1B(\theta)|\right)|\phi_2 - \lambda_2| |e^{v_2}-1|.\label{eq:dot-v4-final}
    \end{align}}%
    \new{The overall Lyapunov functional $V(\eta,\psi)$ as defined in~\eqref{eq:lyapunov-ode-ide} is the sum of $V_4(\eta_1,z)$ and the two terms depending on the functionals $G_i(\psi_i)$. The Dini derivatives of said terms are % Using property~\eqref{eq:dini_g}, the Dini derivatives of said terms are
    \begin{align}
        D^+&\left( \frac{\gamma_i }{\sigma_i} h(G_i(\psi_{i})) \right) \notag \\
        &= \frac{\gamma_i }{\sigma_i} \frac{1}{G_i(\psi_i)}\left(e^{G_i(\psi_i)} - 1\right) D^+\left(G_i(\psi_i)\right) \notag\\
        &\leq -\gamma_i \left(e^{G_i(\psi_i)} - 1\right).\label{eq:intermediate-bound-3}
    \end{align}}%
    \new{Taking \eqref{eq:dot-v4-final} and \eqref{eq:intermediate-bound-3} together,}
    the Dini derivative of \new{the overall Lyapunov functional} $V(\eta,\psi)$~\eqref{eq:lyapunov-ode-ide}, results in
%    \begin{align}
%    D^+ V \leq& 
%    \frac{-4\theta c_1 \lambda_2 \mu(-c_1 \eta_1) - 4c_2 \lambda_2 \mu(z)}{1 + \theta \omega(-c_1 \eta_1) + \omega(z)}\notag \\
%    &+ B(1)|\phi_1 - \lambda_1| |e^{v_1}-1|\notag \\
%    &+ \left(c_1 B(1) + c_1B(\theta)|\right)|\phi_2 - \lambda_2| |e^{v_2}-1|\notag \\
%    & - \gamma_1 (e^{G_1} - 1) - \gamma_2 (e^{G_2} - 1) \\
%    \leq&  \frac{-4\theta c_1 \lambda_2 \mu(-c_1 \eta_1) - 4 c_2 \lambda_2 \mu(z)}{1 + \theta \omega(-c_1 \eta_1) + \omega(z)} \notag \\
%    &+ \left[B(1) |\phi_1 - \lambda_1| - \gamma_1 \right](e^{G_1} - 1) \notag \\
%    & + \left[ \left(c_1 B(1) + c_1B(\theta)\right)|\phi_2 - \lambda_2| - \gamma_2\right] (e^{G_2} - 1) \label{eq:dini-almost-final}\reviewed{,}
%    \end{align}
	\begin{align}
    D^+ V \leq& 
    \frac{-4\theta c_1 \lambda_2 \mu(-c_1 \eta_1) - 4c_2 \lambda_2 \mu(z)}{1 + \theta \omega(-c_1 \eta_1) + \omega(z)} - \gamma_1 (e^{G_1} - 1) \notag \\
    &+ B(1)|\phi_1 - \lambda_1| |e^{v_1}-1| - \gamma_2 (e^{G_2} - 1) \notag \\
    &+ \left(c_1 B(1) + c_1B(\theta)|\right)|\phi_2 - \lambda_2| |e^{v_2}-1| \\
    \leq&  \frac{-4\theta c_1 \lambda_2 \mu(-c_1 \eta_1) - 4 c_2 \lambda_2 \mu(z)}{1 + \theta \omega(-c_1 \eta_1) + \omega(z)} \notag \\
    & + \left[ \left(c_1 B(1) + c_1B(\theta)\right)|\phi_2 - \lambda_2| - \gamma_2\right] (e^{G_2} - 1)  \notag \\
    &+ \left[B(1) |\phi_1 - \lambda_1| - \gamma_1 \right](e^{G_1} - 1) \label{eq:dini-almost-final}\reviewed{,}
    \end{align}
    where \new{the key point lies in bounding $v_i$ through $G_i$ using the fundamental property~\eqref{eq:v_leq_G}}. For notational convenience, the functions $v_i$ and $G_i$ are typeset without the time argument. 
    These expressions cannot be bounded globally.
    For \reviewed{$D^+ V$} to be negative definite at least in a region of state space around the origin, the states are restricted to
    \begin{subequations}
    \begin{align}
        \eta_1 &\leq \ln\left(\frac{\gamma_1}{2\lambda_1 B(1)}\right) =: H_1,\\
        \eta_2 &\leq \ln\left(\frac{\gamma_2}{2\lambda_2 c_1( B(1) + B(\theta))}\right) =: H_2,
    \end{align}\label{eq:eta-restrictions}
    \end{subequations}
    \new{such that
    \begin{align}
        \left[B(1) |\phi_1 - \lambda_1| - \gamma_1\right] &\leq -\frac{\gamma_1}{2}, \\
         \left[\left(c_1 B(1) + c_1B(\theta)\right)|\phi_2 - \lambda_2| - \gamma_2\right] &\leq -\frac{\gamma_2}{2},
    \end{align}
    hold.}
    For $H_1$, $H_2$ to be positive, choose 
        \begin{align}
        \gamma_1 > 2\lambda_1 B(1),\quad
        \gamma_2 > 2\lambda_2 c_1(B(1) + B(\theta)).\label{eq:gamma-temp}
        \end{align}
    \mycomment{After imposing \eqref{eq:gamma-temp}, \eqref{eq:dini-almost-final} becomes
    \begin{align}
        D^+ V \leq& - \frac{4\theta c_1 \lambda_2 \mu(-c_1 \eta_1) - 4c_2 \lambda_2 \mu(z)}{1 + \theta \omega(-c_1 \eta_1) + \omega(z)}
         \notag \\
         &-\frac{\gamma_1}{2}
         \left(e^{G_1}-1\right)
            -\frac{\gamma_2}{2}
            \left(e^{G_2}-1\right) \notag \\
            &\forall \ \eta_1<H_1, \ \eta_2<H_2. \label{eq:Vdot-final}
    \end{align}}
    Imposing \eqref{eq:gamma-temp}, equations \eqref{eq:lyapunov-ode-ide}, \eqref{eq:dini-almost-final} are rewritten as
    %Since $V$ is radially unbounded, all level sets $\Omega_c:=\{\eta \in \mathbb{R}^2, \psi \in \mathcal{S} | V(\eta, \psi)\leq c\}$ are compact. Hence, there exists a level set $c>0$ such that $\Omega_c \subset \mathcal{D}$, where $\mathcal{D}$ is defined by (\ref{eq:constraints}). 
    %Let us now rewrite \eqref{eq:lyapunov-ode-ide}, \eqref{eq:Vdot-final} as
    \begin{align}
    \label{eq:Vdot-final+}
    D^+ V(\eta,G) 
    &\leq - \frac{4\theta c_1 \lambda_2 \mu(-c_1 \eta_1) - 4c_2 \lambda_2 \mu(z)}{1 + \theta \omega(-c_1 \eta_1) + \omega(z)}\notag \\
         %- \mytodo{\frac{\gamma_1}{2}?}  
         &\quad \ -\frac{\gamma_1}{2}
         \left(e^{G_1}-1\right)
           % - \mytodo{\frac{\gamma_2}{2}?} 
            -\frac{\gamma_2}{2}
            \left(e^{G_2}-1\right) 
            \nonumber\\
            & =: - W(\eta,G), \qquad G=(G_1,G_2)
    \end{align}
    for all $\eta_1<H_1,  \eta_2<H_2$. 
    Since $V$ is radially unbounded, all level sets $\Omega_c:=\{\eta \in \mathbb{R}^2, \psi \in \mathcal{S} | V(\eta, \psi)\leq c\}$ are compact. Hence, there exists a level set $c>0$ such that $\Omega_c \subset \mathcal{D}$, where $\mathcal{D}$ is defined by (\ref{eq:constraints}). 

\textbf{Claim 3: Lyapunov $\mathcal{\mathcal{KL}}$ estimates.}
 The functions $V_4$ and $W$ are positive definite in their arguments $(\eta,G)$.
 This is so because $h$ is a positive definite function in its arguments $G_i$ in $V$, functions $e^{G_i}-1$ are positive definite in $G_i$, $\ln(1+V_3)$ is positive definite in $V_3$, the function $V_3(\eta)$ is positive definite in $\eta$ based on the definitions \eqref{eq-omega-mu}, and so is $\frac{4\theta c_1 \lambda_2 \mu(-c_1 \eta_1) - 4c_2 \lambda_2 \mu(z)}{1 + \theta \omega(-c_1 \eta_1)}$ in $W$.
As in \cite{karafyllis2017stability}, it is then concluded that there exists $\theta_0\in\mathcal{KL}$ such that, for all initial conditions $(\eta_0, G_0)$ within the largest $\Omega_c$ contained in $\mathcal{D}$, the estimate \eqref{eq:KL-estimate} holds. 

\new{\textbf{Claim 4: Exponential stability.}
Finally, note that \eqref{eq:lyapunov-ode} is locally quadratic in $\eta$, and that for all $r\in(0,1)$ and for all $\psi_i\in\mathcal{S}$ the property $\|\psi_i\|_\infty \leq r$ holds. Using the definition of functional $G_i$ \eqref{eq:g-functional}, it holds that
\begin{equation}
\|\psi_i\|_\infty \leq G_i(\psi_i) \leq \frac{{\rm e}^{\sigma_i A}}{1-r} \|\psi_i\|_\infty \,,
\end{equation}
and the local asymptotic stability of $\eta=0$,  $\psi=0$ in the norm $\sqrt{\eta_1^2+\eta_2^2} + \|\psi_1\|_\infty+ \|\psi_2\|_\infty$ follows. The exponential nature of asymptotic stability is noted by  examining the Lyapunov estimates and by observing that both $V$ and $W$ are locally quadratic in $\sqrt{\eta_1^2+\eta_2^2} + \|\psi_1\|_\infty+ \|\psi_2\|_\infty$. The locally quadratic dependency of $V$ is the result of the following chain of dependencies: the locally quadratic dependency of $V_3$ on $\eta$, the locally quadratic dependency of $V$ on  $G_i$ through $h$, and the locally linear dependency of $G_i$ on $\|\psi_i\|_\infty$. For similar reasons $W$ is also locally quadratic in $\eta$ and $\|\psi_i\|_\infty$ through the locally linear dependence of $G_i$ on $\|\psi_i\|_\infty$. }
% %a careful inspection of \eqref{eq:Vdot-final+}, 
% the careful inspection of the Jacobian of the closed loop with eigenvalues $s_1=-\lambda_1<0$ and $s_2=-\lambda_2(1+c_1+c_2)<0$,
% along with the definitions of $V, \phi_1,\phi_2, G_1, G_2$. 

\new{To estimate the local exponential decay rate, the linearization of the plant model \eqref{eq:transformed-system} is considered, along with the linearization of the controller \eqref{eq:control-law}, which gives
\rev{
\begin{subequations}
\begin{align}
    \label{eq:doteta_1_linearized}
    \dot\eta_1 &= - \lambda_2(\eta_2+v_2)\\
    \label{eq:doteta_2_linearized}
    \dot \eta_2 &= - \lambda_2\left[(\theta c_1^2 +c_1 c_2)\eta_1 - (c_1+c_2)\eta_2\right]-\lambda_1 v_1\,.
\end{align} \label{eq:doteta_linearized}
\end{subequations}}%
This linear ODE is input-to-state stable (ISS) with respect to the input $(v_1,v_2)$. The $\mathcal{KL}$ component of this linear ODE's ISS bound has the exponential decay rate corresponding to the eigenvalues of the matrix 
$A_\eta =\left[\begin{array}{cc} 0 & -\lambda_2 \\ \rev{\lambda_2(\theta c_1^2 + c_1 c_2)} & -\lambda_2 (c_1+c_2) \end{array}\right]$, which are $\lambda_2 p_1, \lambda_2 p_2$.  Likewise, the nonlinear ODE feedback system \eqref{eq:transformed-system}, \eqref{eq:control-law} is locally ISS with respect to $(v_1,v_2)$, with decay rates reduced by $1+\varepsilon_0$ relative to the linear model, where $\varepsilon_0>0$ is arbitrarily small and this reduced decay rate follows from a Lyapunov argument with a Lyapunov matrix corresponding to $A_\eta$. 
From \eqref{eq:definition-v-map} and \eqref{eq:psi-decays}, we get the ISS relationship $|v_i(t)|\leq N_i ||\psi_{i,t}||_{\infty}\leq N_i M_i e^{-\sigma_i t}||\psi_{i,0}||_{\infty}$, where $N_i>0$ and the exponential decay rates are $\sigma_1,\sigma_2$. Combining the local exponential ISS bound for $(\eta_1,\eta_2)$ with the exponential output stability bound for $(v_1,v_2)$, the decay estimate stated in the theorem results.}
\hfill $\Box$
%This completes the proof of Theorem~\ref{thm:local-stability}.
\end{pf}
\smallskip

\new{On the first encounter of control of age-structured population dynamics, readers and researchers alike often expect that the unmeasured age-dependent functional states $\psi_{i,t}$ need to be estimated by an observer, in order for their estimate to be substituted into the full-state feedback \eqref{eq-controller-original-variables}, instead of the unmeasurable full state $x_i$. But, the $\psi_{i,t}$-dynamics are not observable from any lumped measurements $y_i(t) = \int_0^A q_i(a) x_i(a,t)da$ with positive output kernel functions $q_i$.} 
\rev{However, a more realistically implementable feedback control law is obtained by replacing $\Pi[x_i](t)$ in the controller given in \eqref{eq:control-law} with the $y_i(t)/y_i^*$ where $y_i^* = \int_0^A q_i(a) x_i^*(a) da$. This new control law  renders the equilibrium of the linearized system  exponentially stable.}
\rev{
    \begin{proposition}
        \label{prop:local-exponential-stability}
        The static output-feedback law
%        \begin{multline}
%            u = u^* + \lambda_2\left[c_2\left(1 - \frac{(y_1(t)/y_1^*)^{c_1} }{y_2(t)/y_2^*}\right) \right. \\ \left. + \theta c_1 \left(1 - \left(\frac{y_1(t)}{y_1^*}\right)^{c_1}\right) \right. \\ \left. + \frac{\lambda_1}{\lambda_2}\left(1 - \frac{y_1(t)}{y_1^*}\right) + c_1 \left(\frac{y_2(t)}{y_2^*} - 1\right)\right]
%        \end{multline}
\begin{multline}
\label{eq-static-out-fbk}
u = u^* + \lambda_2\left[c_2\left(1 - \frac{y_1(t)^{c_1} y_2^*}{y_2(t) (y_1^*)^{c_1}}\right) + \frac{\lambda_1}{\lambda_2}\left(1 - \frac{y_1(t)}{y_1^*}\right) \right. \\ \left. + \theta c_1 \left(1 - \left(\frac{y_1(t)}{y_1^*}\right)^{c_1}\right) + c_1 \left(\frac{y_2(t)}{y_2^*} - 1\right)\right]
\end{multline}
renders the {\em linearization} of the closed-loop system with the plant model  \eqref{eq:transformed-system}  exponentially stable in the norm $\sqrt{\eta_1^2+\eta_2^2} + \|\psi_1\|_\infty+ \|\psi_2\|_\infty$.
\end{proposition}
    }

\rev{
\begin{pf}
Substituting the solution \eqref{eq:solution-x} into the lumped measurement results in the relations
\begin{equation}
\frac{y_i(t)}{y_i} = \Pi_i[x_i](t) \left(1 + \delta_i(t)\right)
\end{equation}
where
\begin{equation}
\delta_i(t) = \int_0^A \frac{q_i(a) x_i^*(a)}{y_i^*}\psi(t - a) da
\end{equation}
and $\vert \delta_i(t) \vert \leq A \Vert \psi_i(t) \Vert_{\infty}$. Linearizing the system defined in \eqref{eq:doteta1}, \eqref{eq:doteta2}, in closed loop with control \eqref{eq-static-out-fbk}, around the  equilibrium $\eta_1=\eta_2=0, \psi_1\equiv\psi_2\equiv 0$, results in
\begin{subequations}
\label{eq-eta-lin-out-fbk}
\begin{align} %v_i(\psi_{i,t}) = \ln \left(1 + \int_0^A \bar{b}_j(a) \psi_i(t-a) d a\right)
\dot{\eta}_1 &= -\lambda_2 \left(\eta_2 + \int_0^A \bar{b}_1(a) \psi_2(t-a) da \right)\\
\dot{\eta}_2 &= \lambda_2[(\theta c_1^2 + c_1 c_2)\eta_1 - (c_1 + c_2)\eta_2] \notag \\ &\quad- \lambda_1 \int_0^A \bar{b}_2(a) \psi_1(t-a) da  \notag \\ &\quad+ (\lambda_1 + \lambda_2 (\theta c_1^2 + c_1 c_2)) \delta_1 + \lambda_2(c_1 + c_2) \delta_2
\end{align}
\end{subequations}
For $ \psi_1\equiv\psi_2\equiv 0$, the homogeneous part of this system is governed by the Hurwitz matrix $A_\eta$, as in \eqref{eq:doteta_linearized}, meaning that \eqref{eq-eta-lin-out-fbk} is exponentially ISS with respect to the input $(\psi_1,\psi_2)$. 
By mimicking the proof of Lemma C.4 in \cite{KKK} with $x_1$ replaced by $(\eta_1,\eta_2)$ in \cite[(C.15)]{KKK}, $x_2$ replaced by $(\psi_1,\psi_2)$ in \cite[(C.16)]{KKK}, $u$ set to zero in \cite[(C.15), (C.16)]{KKK}, and the exponential estimate \eqref{eq:psi-decays} used in lieu of the $\mathcal{KL}$ estimate \cite[(C.16)]{KKK}, we arrive at an exponential estimate in the norm $\sqrt{\eta_1^2+\eta_2^2} + \|\psi_1\|_\infty+ \|\psi_2\|_\infty$ in the same manner as the $\mathcal{KL}$ estimate \cite[(C.17)]{KKK} is obtained. 
%Thus, the local exponential stability follows the same arguments as presented in the proof of Theorem~\ref{thm:local-stability} for Claim 4.
\hfill$\Box$\end{pf}
}
\rev{The consequence of Proposition~\ref{prop:local-exponential-stability} is that, as long as the output kernel functions $q_i(a)$ and the equilibrium profiles $x_i^\ast(a)$ are known, an observer is not needed (nor is it usable, due to the lack of observability of $\psi_i$). The substitution of $\Pi_i(x_i)$ by $y_i(t)/y_i^*$ of course affects the estimate \eqref{eq:Omegac},
\eqref{eq:constraints} of the region of attraction but the linearization is exponentially stable.
}

\rev{
Proposition~\ref{prop:local-exponential-stability} establishes the exponential stability of the closed-loop system in the $(\eta_1,\eta_2,\psi_1,\psi_2)$ representation. But how about the linearization of the closed-loop system  
\eqref{eq:sys}, \eqref{eq-static-out-fbk} in the original $(x_1,x_2)$ variables, at the equilibrium $(x_1^*, x_2^*)$? One can establish the exponential stability (e.s.) of the linearization using, for the left side of the e.s. estimate, the relation \eqref{eq:solution-x} and approximating it as $x_i(a,t)-x_i^*(a) \approx x_i^*(0)(\eta_i(t)+\psi_i(t-a))$, and, for the right side of the e.s. estimate, the relations \eqref{eq:system-trafo} linearly approximated in terms of the initial conditions $x_i(a,0)-x_i^*(a)$. However, the latter step leads to complicated expressions on the e.s. estimate's right-hand side, and we forego the development of the claim of the e.s. of the linearization of 
\eqref{eq:sys}, \eqref{eq-static-out-fbk}. 
}

% \begin{pf}
% \new{Instead of estimating $\psi_{i,t}$, from \eqref{eq:solution-x} it follows that 
% \begin{eqnarray}
% \Pi_i(x_i)(t) &=& \ln\frac{y_i(t)}{\int_0^A q_i(a) x_i^\ast(a)da} + \delta_i(t)
% \\
% \delta_i(t) &=& - \ln\left( 1+ \sigma_i(t)\right)
% \\
% \sigma_i(t) &=& \int_0^A \frac{ q_i(a) x_i^\ast(a)}{\int_0^A q_i(\alpha) x_i^\ast(\alpha)d\alpha}\psi_i(t-a) da
% \\  \label{eq-delta-Gi}
% |\delta_i(t)| &\leq & G_i(t)\,.
% \end{eqnarray}
% Due to the exponentially decaying, $G_i$-bounded nature of $\delta_i$, the terms $\Pi_i(x_i)$ in the feedback law \eqref{eq-controller-original-variables} can be replaced by the signals $\ln\frac{y_i(t)}{\int_0^A q_i(a) x_i^\ast(a)da}$, which employ the output measurements $y_i(t)$ and where the error of this substitution is bounded by the exponentially decaying $\delta_i$. As long as the output kernel functions $q_i(a)$ and the equilibrium profiles $x_i^\ast(a)$ are known, an observer is not needed (nor is it usable, due to the lack of observability of $\psi_i$). The substitution of $\Pi_i(x_i)$ by $\ln\frac{y_i(t)}{\int_0^A q_i(a) x_i^\ast(a)da}$ of course affects the estimate \eqref{eq:Omegac},
% \eqref{eq:constraints} of the region of attraction. But local asymptotic stability holds, due to \eqref{eq-delta-Gi} and \eqref{eq:dini_g}. }
% \end{pf}

\begin{figure}
    \centering
    \includegraphics{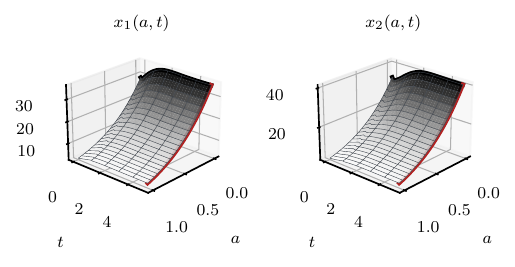}
    \includegraphics{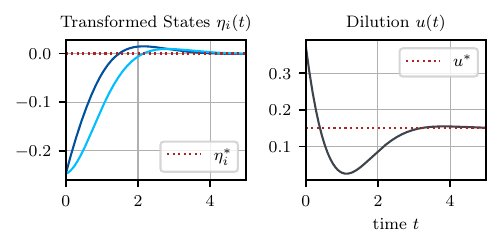}
    \caption{Simulation results of the population densities $x_i$ from system (\ref{eq:sys}) and the transformed state variables $\eta_1$ (\mlLineLegend{dunkelblau}) and $\eta_2$ (\mlLineLegend{hellblau})
    %from representation (\ref{eq:transformed-system}) under control law (\ref{eq:control-law}), and with parameter set (\ref{eq:parameters}), ICs (\ref{eq:ic-1}), $\theta=c_1=c_2=1$. Simulations of the $\psi$-dynamics are omitted as they are stable.
    }
    \label{fig:simulation-results}
\end{figure}

\section{Simulations}\label{sec:simulations}
The age-dependent kernels for simulation are 
\begin{align}
	\mu_i(a)=\bar{\mu}_i\mathrm{e}^a, \  k_i(a)=\bar{k}_i\mathrm{e}^{-a}, \ 
	b_i(a)=\bar{b}_i\left(a-a^2\right), \label{eq:parameters}
\end{align}
which are biologically inspired \new{from piror studies on E. coli populations in chemostats} \cite{kurth2023control}. 
% \new{The exponential form of the birth rate used is motivated by prior experimental and modeling studies on E. coli populations in chemostats \cite{kurth2023control}, where the age structure is governed by continuous growth and dilution processes. 
% In such microbial systems, individuals reproduce continuously and experience a constant probability of dilution or death over age, which naturally leads to an exponentially decreasing birth rate profile. 
% }
In this example, the maximum age is set to $A=1$ and both species to exhibit the same behavior, namely $\bar{\mu}_i = 0.5$, $\bar{k}_i = 3$, $\bar{b}_i =0.4$.
To demonstrate the stabilization capacity of the designed feedback law, the challenging case of ICs when both populations are slightly underpopulated are chosen, namely
\begin{align}
    x_1(0) = x_1^*(a) e^{-0.2(1+a)} , \ x_2(0) = x_2^*(a) e^{-0.2(1+a)} \label{eq:ic-1}
\end{align}
such that $\eta_0=[-0.25, -0.25]^\top$.
Figure~\ref{fig:simulation-results} shows simulations of system~(\ref{eq:sys}) with control (\ref{eq:control-law}), parameter set (\ref{eq:parameters}), ICs (\ref{eq:ic-1}).
The parameters of the control law (\ref{eq:control-law}) are chosen to be $\theta=c_1=c_2=1$.
The control achieves convergence to the desired equilibrium while remaining positive.

\section{Conclusion}\label{sec:conclusion}
\rev{Building on foundational work \cite{karafyllis2017stability}, this paper addresses exponentially unstable predator-predator dynamics, achieving local exponential stabilization via backstepping. Future work will explore extending these methods to general 
$n$-species models with diverse competition types, opening promising directions for applications in epidemiology and ecosystem stability, including food webs with multiple predators, preys, and complex interactions.}

\begin{ack}                     % Place acknowledgements
This work was partially supported by the German Research Foundation (Deutsche Forschungsgemeinschaft) under grant SA 847/22-2, project number 327834553.
\end{ack}

\bibliographystyle{plain}        % Include this if you use bibtex 
\bibliography{main}           % and a bib file to produce the 
% bibliography (preferred). The
% correct style is generated by
% Elsevier at the time of printing.

% \appendix
% \section{A summary of Latin grammar}    % Each appendix must have a short title.
% \section{Some Latin vocabulary}         % Sections and subsections are supported  
%                                         % in the appendices.
\end{document}